\documentclass[a4paper,envcountsame]{llncs}

\makeatletter
\renewcommand\section{\@startsection{section}{1}{\z@}%
                       {-12\p@ \@plus -2\p@ \@minus -2\p@}%
                       {6\p@ \@plus 2\p@ \@minus 2\p@}%
                       {\normalfont\large\bfseries\boldmath
                        \rightskip=\z@ \@plus 8em\pretolerance=10000 }}
\makeatother

\usepackage{enumerate,amsfonts,graphicx,url,stfloats,subfig}
\usepackage{amssymb,amsmath,color,verbatim,bm}
\usepackage{stmaryrd}
\usepackage{paralist}
\usepackage{environ,xspace}
\usepackage{thmtools,thm-restate}

\usepackage{tikz}
\usetikzlibrary{arrows,automata,shapes}

\usepackage[color=light gray]{todonotes}
\presetkeys{todonotes}{inline}{}

\NewEnviron{killcontents}{}

\DeclareMathAlphabet{\mathpzc}{OT1}{pzc}{m}{it}

\newcommand{\ie}{{i.e.}\xspace}

\newcommand{\etc}{{etc.}\xspace}

\newcommand{\lemref}[1]{Lemma~\ref{lem:#1}}

\newcommand{\defref}[1]{Definition~\ref{def:#1}}
\newcommand{\secref}[1]{Sec.~\ref{sec:#1}}

\newcommand{\exref}[1]{Example~\ref{ex:#1}}

\newcommand{\thmref}[1]{Theorem~\ref{thm:#1}}
\newcommand{\propref}[1]{Proposition~\ref{prop:#1}}

\newcommand{\ignore}[1]{}

\newcommand{\s}{s}   
\newcommand{\str}{t}
\newcommand{\sStr}{u}

\newcommand{\hd}{d_{M}}

\newcommand{\diff}{\mathtt{diff}}
\newcommand{\inp}{w}   
\newcommand{\inpt}{v}   
\newcommand{\out}{w'} 
\newcommand{\outt}{v'} 
\newcommand{\Nat}{\mathbb{N}}

\newcommand{\funcDefinedBy}[1]{\ensuremath{\llbracket #1\rrbracket}}

\newcommand{\eos}{\texttt{\#}}

\newcommand{\eps}{\epsilon}

\newcommand{\tran}{\ensuremath{\mathpzc{T}}}
\newcommand{\trans}{\ensuremath{\mathpzc{T}}}

\newcommand{\ptran}{\ensuremath{\mathpzc{P}}_\tran}

\newcommand{\Pair}{\mathpzc{P}}

\newcommand{\es}{\eps}

\tikzstyle{smalltext}=[font=\fontsize{7}{8}\selectfont]
\tikzstyle{captiontext}=[font=\fontsize{8}{8}\selectfont]
\tikzstyle{state}=[draw, ellipse, minimum height=5mm,
                   minimum width=5mm, inner sep=1pt,
                   text=black,
                   font=\fontsize{8}{8}\selectfont,
                   semithick]

\newcommand{\lnorm}{Manhattan distance\xspace}

\newcommand{\fst}{{\sc fst}\xspace}
\newcommand{\wa}{{\sc wa}\xspace}

\newcommand{\conp}{\textsc{co-NP}}
\newcommand{\np}{\textsc{NP}}

\newcommand{\lpair}[2]{\langle {#1}, {#2} \rangle}

\newcommand{\stutter}{\textsc{Stutter}}

\newcommand{\alphI}{\Sigma}
\newcommand{\alphO}{\Gamma}

\newcommand{\baut}{\bar{\aut}}
\newcommand{\blaut}{\baut^L}
\newcommand{\braut}{\baut^R}

\newcommand{\dI}{d_{\alphI}}
\newcommand{\dO}{d_{\alphO}}
\newcommand{\dsetO}{D_{\alphO}}

\newcommand{\PCPinst}{{\cal G}}
\newcommand{\diffI}{\diff_{\alphI}}
\newcommand{\diffO}{\diff_{\alphO}}

\newcommand{\dom}[1]{\mathrm{dom}(#1)}

\newcommand{\N}{{\cal N}}

\newcommand{\aut}{{\cal A}}

\newcommand{\Q}{{\mathbb{Q}}}
\newcommand{\cost}{{c}}
\newcommand{\lang}{{\cal L}}
\newcommand{\Lfin}{{\cal L}_{\mathrm{fin}}}

\newcommand{\Acc}{\mathsf{Acc}}
\newcommand{\const}{\lambda}

\newcommand{\fsum}{\textsc{Sum}}

\newcommand{\fraction}[1]{\underline{#1}}
\newcommand{\fdiscF}{\textsc{Disc}_{\lambda}^{\textrm{fin}}}
\newcommand{\fdiscI}{\textsc{Disc}_{\lambda}^{\textrm{inf}}}
\newcommand{\fdisc}{\textsc{Disc}_{\delta}}
\newcommand{\flimavg}{\textsc{LimAvg}}

\newcommand{\valFun}{\textsc{ValFunc}}

\newcommand{\valueL}[1]{{\cal L}_{{#1}}}

\newcommand{\lett}[1]{\underline{#1}}

\newcommand{\Qinf}{\Q^{\infty}}

\title{Lipschitz Robustness of Finite-state Transducers}
\author{Thomas A. Henzinger, Jan Otop, \and Roopsha Samanta}
\institute{IST Austria \\
\email{\{tah,jotop,rsamanta\}@ist.ac.at}}

\begin{document}


\maketitle

\begin{abstract}
We investigate the problem of checking if a finite-state transducer is {\em
robust} to uncertainty in its input. Our notion of robustness is based on the
analytic notion of Lipschitz continuity --- a transducer is $K$-(Lipschitz) 
robust if the perturbation in its output is at most $K$
times the perturbation in its input.  We quantify input and output perturbation
using {\em similarity functions}.  We show that $K$-robustness is undecidable
even for deterministic transducers.  We identify a class of functional
transducers, which admits a polynomial time 
automata-theoretic decision procedure for 
$K$-robustness. This class includes Mealy machines and functional
letter-to-letter transducers.  We also study $K$-robustness of nondeterministic
transducers.  Since a nondeterministic transducer generates a set of output
words for each input word, we quantify output perturbation using {\em
set-similarity functions}.  We show that $K$-robustness of nondeterministic
transducers is undecidable, even for letter-to-letter transducers. 
We identify a class of set-similarity functions which admit decidable 
$K$-robustness of letter-to-letter transducers.

\end{abstract}

\section{Introduction}\label{sec:intro}
Most computational systems today are embedded in a physical environment. The
data processed by such real-world computational systems is often noisy or
uncertain.  For instance, the data generated by sensors in reactive systems such 
as avionics software may be corrupted, keywords processed by text processors may be wrongly spelt,
the DNA strings processed in computational biology may be incorrectly sequenced, 
and so on. In the presence of such input uncertainty, it is not enough for a
computational system to be functionally correct. An additional desirable
property is that of {\em continuity} or {\em robustness} --- the system behaviour 
degrades smoothly in the presence of input disturbances \cite{Henzinger08}.

Well-established areas within control theory, such as {\em robust control}
\cite{ZDG96}, extensively study robustness of systems. However, their results
typically involve reasoning about continuous state-spaces and are not directly
applicable to inherently discontinuous discrete computational systems.
Moreover, uncertainty in robust control refers to differences between
a system's model and the actual system; thus robust control focuses
on designing controllers that function properly in the presence of perturbation
in various internal parameters of a system's model. Given the above, formal reasoning
about robustness of computational systems under input uncertainty is a problem
of practical as well as conceptual importance.  

In our work, we focus on robustness of finite-state transducers, processing 
finite or infinite words, in the presence
of uncertain inputs. Transducers are popular models of input-output
computational systems operating in the real world
\cite{gusfield_algorithms_1997,mohri1997finite,BradleyH07,VHLMB12}.  While many
decision problems about transducers have been studied thoroughly over the
decades \cite{mohri1997finite,VHLMB12}, their behaviour under uncertain inputs
has only been considered recently \cite{foundation}. In \cite{foundation}, a
transducer was defined to be robust if its output changed proportionally to
every change in the input {\em upto a certain threshold}. In practice, it may
not always be possible to determine such a bound on the input perturbation.
Moreover, the scope of the work in \cite{foundation} was limited to the
robustness problem for {\em functional} transducers w.r.t. specific
distance functions, and did not consider arbitrary nondeterministic transducers
or arbitrary {\em similarity functions}.

In this paper, we formalize robustness of finite-state transducers as Lipschitz
continuity. A function is Lipschitz-continuous if its output changes
proportionally to {\em every} change in the input.  Given a constant $K$ and
similarity functions $\dI$, $\dO$ for computing the input, output
perturbation, respectively, a functional transducer $\tran$ is defined to be
$K$-Lipschitz robust (or simply, $K$-robust) w.r.t. $\dI$, $\dO$ if for all
words $\s,\str$ in the domain of $\tran$ with finite $\dI(\s,\str)$, 
$\dO(\tran(\s),\tran(\str)) \le K \dI(\s,\str)$.  Let us consider the
transducers $\tran_{NR}$ and $\tran_R$ below.  Recall that the Hamming distance
between equal length words is the number of positions in which the words
differ.  Let $\dI$, $\dO$ be computed as the Hamming distance for equal-length
words, and be $\infty$ otherwise. Notice that for words $a^{k+1}$, $ba^{k}$ in
the domain of the Mealy machine $\tran_{NR}$, 
$\dI(a^{k+1}, ba^{k}) = 1$ and the distance between the corresponding 
output words, $\dO(a^{k+1}, b^{k+1})$, equals $k+1$. 
Thus, $\tran_{NR}$ is not $K$-robust for any $K$. On the other hand, the
transducer $\tran_R$ is $1$-robust: for words $a^{k+1}$, $ba^{k}$, we have 
$\dI(a^{k+1}, ba^{k}) = \dO((b)^{k+1},a(b)^k) = 1$, 
and for all other words $\s,\str$ in the domain of
$\tran_R$, either $\dI(\s,\str) = \infty$ or $\dI(\s,\str) =
\dO(\tran_R(\s),\tran_R(\str)) = 0$.


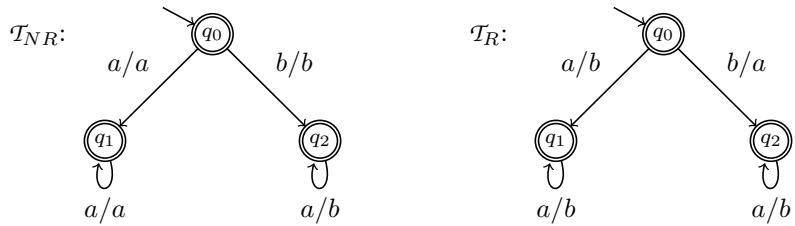
\begin{center}
\begin{minipage}{0.45\linewidth}
\begin{tikzpicture}[->,node distance=2cm,semithick]
  \tikzstyle{every state}=[circle,text=black]

  \node[state,accepting]   (q0)                     {$q_0$};
  \node[state,accepting] (q1) [below left of=q0]  {$q_1$};
  \node[state,accepting] (q2) [below right of=q0] {$q_2$};

  \node[state,draw=none] (name) [left of=q0,node distance=2.3cm]  {\normalsize{$\tran_{NR}$:}};

  \node[coordinate] (c1) [above left of=q0,node distance=0.5cm,xshift=-3mm] {};

  \path (c1) edge (q0)
        (q0) edge [above left] node {$a/a$} (q1)
        (q0) edge [above right] node {$b/b$} (q2)
        (q1) edge [loop below] node {$a/a$} (q1)
        (q2) edge [loop below] node {$a/b$} (q2);
\end{tikzpicture}
\end{minipage}
\begin{minipage}{0.45\linewidth}
\centering
\begin{tikzpicture}[->,node distance=2cm,semithick]
  \tikzstyle{every state}=[circle,fill=grey,text=black]

  \node[state,accepting]   (q0)                     {$q_0$};
  \node[state,accepting] (q1) [below left of=q0]  {$q_1$};
  \node[state,accepting] (q2) [below right of=q0] {$q_2$};
  \node[state,draw=none] (name) [left of=q0,node distance=2.3cm]  {\normalsize{$\tran_{R}$:}};
  \node[coordinate] (c1) [above left of=q0,node distance=0.5cm,xshift=-3mm] {};

  \path (c1) edge (q0)
        (q0) edge [above left] node {$a/b$} (q1)
        (q0) edge [above right] node {$b/a$} (q2)
        (q1) edge [loop below] node {$a/b$} (q1)
        (q2) edge [loop below] node {$a/b$} (q2);
\end{tikzpicture}
\end{minipage}
\end{center}

While the $K$-robustness problem is undecidable even for deterministic
transducers, we identify interesting classes of finite-state transducers with  
decidable $K$-robustness. We first define a class of functional transducers, 
called synchronized transducers, which admits a polynomial time 
decision procedure for $K$-robustness. 
This class includes Mealy
machines and functional letter-to-letter transducers; membership of a
functional transducer in this class is decidable in polynomial time.  Given
similarity functions computable by weighted automata, we reduce the
$K$-robustness problem for synchronized transducers to the emptiness problem
for weighted automata. 

We  extend our decidability results by employing an \emph{isometry approach}.
An \emph{isometry} is a transducer, which for all words $\s,\str$ satisfies
$\dO(\trans(\s),\trans(\str)) = \dI(\s,\str)$. 
We observe that if a transducer $\trans_2$ can be obtained from a transducer $\trans_1$ by applying isometries to the
input and output of $\trans_1$, then $K$-robustness of $\trans_1$ and $\trans_2$ coincide.
This observation enables us to reduce  $K$-robustness of various 
transducers to that of synchronized transducers. 

Finally, we study $K$-robustness of nondeterministic transducers.  Since a
nondeterministic transducer generates a set of output words for each input
word, we quantify output perturbation using {\em set-similarity functions} and
define $K$-robustness of nondeterministic transducers w.r.t. such
set-similarity functions. We show that $K$-robustness of nondeterministic
transducers is undecidable, even for letter-to-letter transducers.  We define
three classes of set-similarity functions and show decidability of
$K$-robustness of nondeterministic letter-to-letter transducers w.r.t. one
class of set-similarity functions. 


The paper is organized as follows. We begin by presenting necessary definitions
in \secref{prelim}.  We formalize Lipschitz robustness in \secref{probdef}. In
\secref{regular} and \secref{functional}, we study the $K$-robustness problem
for functional transducers, showing undecidability of the general problem and
presenting two classes with decidable $K$-robustness. We study $K$-robustness
of arbitrary nondeterministic transducers in \secref{nondeter}, present 
a discussion of related work in \secref{related} and conclude in \secref{conclusion}.

\section{Preliminaries}\label{sec:prelim}
\newcommand{\alphIi}{{\Lambda}}
\newcommand{\alphOo}{{\Delta}}
\newcommand{\dIi}{d_{\alphIi}}
\newcommand{\dOo}{d_{\alphOo}}
\newcommand{\run}{\gamma}

In this section, we review definitions of finite-state transducers and weighted
automata, and present similarity functions. We use the following notation.  We
denote input letters by $a$, $b$ \etc, input words by $\s$, $\str$ \etc, output
letters by $a'$, $b'$ \etc and output words by $\s'$, $\str'$ \etc We denote
the concatenation of words $\s$ and $\str$ by $\s \cdot \str$, the $i^{th}$
letter of word $\s$ by $\s[i]$, the subword $\s[i] \cdot \s[i+1] \cdot \ldots
\cdot \s[j]$ by $\s[i,j]$, the length of the word $\s$ by $|\s|$, and the empty
word and empty letter by $\es$. Note that for an $\omega$-word $\s$, $|\s| = \infty$.
\medskip 

\noindent {\bf Finite-state Transducers.}\label{subsec:fst} A finite-state
transducer (\fst) $\tran$ is given by a tuple $(\alphI,\alphO,Q,Q_0,E,F)$ where
$\alphI$ is the input alphabet, $\alphO$ is the output alphabet, $Q$ is a
finite nonempty set of states, $Q_0 \subseteq Q$ is a set of initial states, $E
\subseteq Q \times \alphI \times \alphO^* \times Q$ is a set of
transitions\footnote{Note that we disallow $\es$-transitions where the
transducer can change state without moving the reading head.}, and $F$ is a 
set of accepting states.

A run $\run$ of $\tran$ on an input word $\s = \s[1]\s[2]\ldots$ is defined in
terms of the sequence: $(q_0,\out_1)$, $(q_1,\out_2)$, $\ldots$ where $q_0 \in
Q_0$ and for each $i \in \{1,2,\ldots\}$, $(q_{i-1},\s[i],\out_i,q_{i}) \in E$. 
Let $\text{Inf}(\run)$ denote the set of states that appear infinitely often along $\run$.
For an \fst $\tran$
processing $\omega$-words, a run is accepting if $\text{Inf}(\run)
\cap F \neq \emptyset$ (B{\"u}chi acceptance condition). 
For an \fst $\tran$ processing finite words, a run
$\run$: $(q_0,\out_1)$, $\ldots$ $(q_{n-1},\out_n)$, $(q_n,\es)$ on 
input word $\s[1]\s[2]\ldots\s[n]$ is accepting if $q_n \in F$ (final state acceptance 
condition). 
The output of $\tran$ along a run is the word 
$\out_1 \cdot \out_2 \cdot \ldots$ 
if the run is accepting, and is undefined otherwise. The transduction computed
by an \fst $\tran$ processing infinite words (resp., finite words) 
is the relation $\funcDefinedBy{\tran} \subseteq \alphI^\omega 
\times \alphO^\omega$ (resp., $\funcDefinedBy{\tran} \subseteq \alphI^*  
\times \alphO^*$), where $(\s,\s') \in \funcDefinedBy{\tran}$ iff there is an
accepting run of $\tran$ on $\s$ with $\s'$ as the output along that run. With
some abuse of notation, we denote by $\funcDefinedBy{\tran}(\s)$ the set
$\{\str: (\s,\str) \in \funcDefinedBy{\tran}\}$. The input language,
$\dom{\tran}$, of $\tran$ is the set $\{\s:
\funcDefinedBy{\tran}(\s)\ \text{is non-empty}\}$.

An \fst $\tran$ is called {\em functional} if the
relation $\funcDefinedBy{\tran}$ is a function. In this case, we use
$\funcDefinedBy{\tran}(\s)$ to denote the unique output word generated  along
any accepting run of $\tran$ on input word $\s$.  Checking if an arbitrary \fst
is functional can be done in polynomial time \cite{gurari_note_1983}. An 
\fst $\tran$ is {\em deterministic} if for each $q \in Q$ and each $a \in \alphI$, 
$|\{q': (q,a,\out,q') \in E\}| \leq 1$. An \fst
$\tran$ is a {\em letter-to-letter} transducer if for every transition of the
form $(q,a,\out,q') \in E$, $|\out| = 1$. A {\em Mealy machine} is 
a deterministic, letter-to-letter transducer, with every state being an accepting state. In what follows, we use transducers and finite-state transducers interchangeably.

\noindent{\em Composition of transducers.} Consider transducers $\trans_1 =
(\alphI, \alphOo, Q_1, Q_{1,0}, E_1, F_1)$ and  $\trans_2 = (\alphOo,
\alphO, Q_2, Q_{2,0}, E_2, F_2)$ such that for every $\s \in
\dom{\trans_1}$, $\funcDefinedBy{\trans}(\s) \in \dom{\trans_2}$.  We define
$\trans_2 \circ \trans_1$, the \emph{composition} of $\trans_1$ and $\trans_2$,
as the transducer $(\alphI,\alphO,Q_1 \times Q_2, Q_{1,0} \times
Q_{2,0},E,F_1 \times F_2)$, where $E$ is defined as:
$(\lpair{q_1}{q_2}, a, \out, \lpair{q_1'}{q_2'})$ $\in$ $E$ iff $(q_1, a,
\str', q_1') \in E_1$ and upon reading $\str'$, $\trans_2$ generates $\out$ and 
changes state from $q_2$ to $q_2'$, \ie, iff $(q_1, a,
\str', q_1') \in E_1$ and there exist
$(q_2,\str'[1],\out_1,q_2^{1})$, $(q_2^1,\str'[2],\out_2,q_2^{2})$, $\ldots$, 
$(q_2^{k-1}, \str'[k],\out_k,q_2') \in E_2$ such that $k = |\str'|$ and 
$\out = \out_1 \cdot \out_2 \cdot \ldots \cdot \out_k$.  Observe that if $\trans_1, \trans_2$ are functional,
$\trans_2 \circ \trans_1$ is functional and $\funcDefinedBy{\trans_2
\circ \trans_1} = \funcDefinedBy{\trans_2} \circledcirc \funcDefinedBy{\trans_1}$, where  
$\circledcirc$ denotes function composition.
\medskip

%
%
\noindent{\bf Weighted automata.} Recall that a finite automaton (with B{\"u}chi or 
final state acceptance) can be expressed as a 
tuple $(\alphI, Q, Q_0, E, F)$, where $\alphI$ is the alphabet, $Q$ is a finite
set of states, $Q_0 \subseteq Q$ is a set of initial states, $E \subseteq Q
\times \alphI \times Q$ is a transition relation, and $F \subseteq Q$ is a set
of accepting states.  A weighted automaton (\wa) is a finite automaton whose transitions
are labeled by rational numbers.  Formally, a \wa $\aut$ is a tuple $(\alphI,
Q, Q_0, E, F, \cost)$ such that $(\alphI, Q, Q_0, E, F)$ is a finite automaton
and $\cost : E \mapsto \Q$ is a function labeling the transitions of $\aut$. 
The transition labels are called \emph{weights}.

Recall that a run $\pi$ of a finite automaton on a word $\s =
\s[1]\s[2]\ldots$ is defined as a sequence of states: $q_0, q_1, \ldots$ 
where $q_0 \in Q_0$ and for each  $i \in \{1,2,\ldots\}$, $(q_{i-1}, s[i],
q_{i}) \in E$.  A run $\pi$ in a finite automaton processing 
$\omega$-words (resp., finite words) is accepting if it satisfies 
the B{\"u}chi (resp., final state) acceptance condition. 
The set of accepting
runs of an automaton on a word $\s$ is denoted $\Acc(\s)$. Given a word $\s$,
every run $\pi$ of a \wa $\aut$ on $\s$ defines a sequence
$\cost(\pi)=(\cost(q_{i-1}, \s[i], q_{i}))_{1\leq i \leq |\s|}$ of weights of
successive transitions of $\aut$; such a sequence is also referred to as a
weighted run.  To define the semantics of weighted automata we need to define
the value of a run (that combines the sequence of weights of the run into a
single value) and the value across runs (that combines values of different runs
into a single value).  To define values of runs, we consider \emph{value
functions} $f$ that assign real numbers to sequences of rational numbers, and
refer to a \wa with a particular value function $f$ as an $f$-\wa.  Thus, the
value $f(\pi)$ of a run $\pi$ of an $f$-\wa $\aut$ on a word $\s$ equals $f(\cost(\pi))$.
The value of a word $\s$ assigned by an $f$-\wa $\aut$, denoted
$\valueL{\aut}(\s)$, is the {\em infimum} of the set of values of all accepting
runs, i.e., $\valueL{\aut}(\s) = \inf_{\pi \in \Acc(\s)} f(\pi)$ (the
infimum of an empty set is infinite).  

In this paper, we consider the following value
functions: (1)~the sum function $\fsum(\pi) = \sum_{i=1}^{|\pi|}
(\cost(\pi))[i]$, (2)~the discounted sum function $\fdisc(\pi) =
\sum_{i=1}^{|\pi|} \delta^i (\cost(\pi))[i]$ with $\delta \in (0,1)$ 
and (3)~the limit-average
function $\flimavg(\pi) = \limsup_{k\rightarrow \infty} \frac{1}{k}
\sum_{i=1}^{k} (\cost(\pi))[i]$. Note that the limit-average value function cannot be 
used with finite sequences. We define $\valFun = \{ \fsum, \fdisc, \flimavg\}$.

A \wa $\aut$ is \emph{functional} iff for every word $\s$,
all accepting runs of $\aut$ on $\s$ have the same value. 

%
%

\noindent{\em Decision questions. }Given an $f$-\wa $\aut$ and a threshold $\const$, the \emph{emptiness} question
asks whether {\em there exists} a word $\s$ such that  $\valueL{\aut}(\s) <
\const$ and the \emph{universality} question asks whether {\em for all} words
$\s$ we have $\valueL{\aut}(\s) < \const$.  The following results are known.

\begin{lemma}
(1)~For every $f \in \valFun$,
the emptiness problem is decidable 
in polynomial time for nondeterministic $f$-automata
~\cite{Filar:1996:CMD:248676,Droste:2009:HWA:1667106}.
(2)~The universality problem is undecidable for $\fsum$-automata with 
weights drawn from $\{ -1, 0, 1\}$~\cite{DBLP:journals/ijac/Krob94,DBLP:conf/atva/AlmagorBK11}.
\label{lem:oldResults}
\end{lemma}


\noindent{\em Remark}.
Weighted automata have been defined over semirings~\cite{Droste:2009:HWA:1667106} as
well as using value functions (along with infimum or supremum) as above~\cite{DBLP:conf/fct/ChatterjeeDH09,DBLP:journals/tocl/ChatterjeeDH10}.                 
These variants of weighted automata have incomparable expression power.
We use the latter definition as it enables us to express long-run average and discounted sum,
which are inexpressible using weighted automata over semirings.
Long-run average and discounted sum are
widely used in quantitative verification and define natural distances (\exref{picewise-linear}).
Moreover, unlike the semiring-based definition, the value-function-based definition 
extends easily from finite to infinite words.\\
%

\noindent{\bf Similarity Functions.} In our work, we use similarity functions
to measure the similarity between words. Let $\Qinf$ denote the set $\Q \cup \{\infty\}$. 
A similarity function $d:S \times S\to \Qinf$ is a function with the properties: $\forall x,y \in S:$ (1) $d(x,y) \ge 0$ and (2)
 $d(x,y) = d(y,x)$. A similarity function $d$ is also a
distance (function or metric) if it satisfies the additional
properties: $\forall x,y,z \in S:$ (3) $d(x,y) = 0$ iff $x = y$ and (4)  $d(x,z)
\le d(x,y) + d(y,z)$. We emphasize that in our work we do not need to restrict similarity functions to be distances.  

An example of a similarity function
is the \emph{generalized Manhattan distance} defined as: 
$\hd(\s,\str) = \sum_{i=1}^{\infty} \mathtt{diff}(\s[i],\str[i])$ 
for infinite words $\s, \str$ (resp.,  
$\hd(\s,\str) = \sum_{i=1}^{max(|\s|,|\str|)} \mathtt{diff}(\s[i],\str[i])$ 
for finite $\s, \str$), where 
$\mathtt{diff}$ is the mismatch penalty for substituting letters.
The mismatch penalty is required to be a distance function 
on the alphabet (extended with a special end-of-string letter $\eos$ for finite words). When $\mathtt{diff}(a,b)$ is defined to be $1$ for all $a,b$ with $a \neq b$, and $0$ otherwise, 
$\hd$ is called the \emph{Manhattan distance}. 
 
\noindent {\em Notation}: We use $\s_1 \otimes \ldots \otimes \s_k$ to denote {\em
convolution} of words $\s_1, \ldots, \s_k$, for $k >1$. The convolution of $k$
words merges the arguments into a single word over a $k$-tuple alphabet 
(accommodating arguments of different lengths using $\eos$ letters at the ends
of shorter words).  Let $\s_1, \ldots, \s_k$ be words over alphabets $\Sigma_1,
\ldots, \Sigma_k$.  Let $\Sigma_1 \otimes \ldots \otimes \Sigma_k$ denote the
$k$-tuple alphabet $(\Sigma_1 \cup \{\eos\}) \times \ldots \times (\Sigma_k
\cup \{\eos\})$.  The convolution $\s_1 \otimes \ldots \otimes \s_k$ is an 
infinite word (resp., a finite word of length $max(|\s_1|, \ldots, |\s_k|)$), 
over $\Sigma_1 \otimes \ldots \otimes
\Sigma_k$, such that: for each $i \in \{ 1, \ldots, |\s_1 \otimes \ldots
\otimes \s_k|\}$, $(\s_1 \otimes \ldots \otimes \s_k)[i] = \langle \s_1[i],
\ldots, \s_k[i] \rangle$ (with $\s_j[i] = \eos$ if $i > |\s_j|$).  For example,
the convolution $aa \otimes b \otimes add$ is the 3 letter word $\langle a,b,a
\rangle \langle a,\eos,d \rangle \langle \eos, \eos, d \rangle$.

\begin{definition}[Automatic Similarity Function]
A similarity function $d : \Sigma_1^\omega \times \Sigma_2^\omega \mapsto \Q$ is called
automatic if there exists a \wa $\aut_d$ over $\Sigma_1
\otimes \Sigma_2$ such that $\forall \s_1 \in
\Sigma_1^\omega,\s_2 \in \Sigma_2^\omega$: $d(\s_1,\s_2) = \lang_{\aut_d}(\s_1
\otimes \s_2)$. We say that $d$ is computed by $\aut_d$.
\end{definition}

One can similarly define automatic similarity functions over finite words.

\section{Problem Definition}\label{sec:probdef}
Our notion of robustness for transducers is based on the analytic notion of
Lipschitz continuity. We first define $K$-Lipschitz robustness of functional transducers.


\begin{definition}[$K$-Lipschitz Robustness of Functional Transducers]
\label{def:robust}
Given 
a constant $K \in \Q$ with $K > 0$ and similarity functions 
$\dI:\alphI^\omega \times \alphI^\omega \; \to \; \Qinf$ 
(resp., $\dI:\alphI^* \times \alphI^* \; \to \; \Qinf$)  
and $\dO: \alphO^\omega \times \alphO^\omega \; \to \; \Qinf$ (resp., $\dO: \alphO^* \times \alphO^* \; \to \; \Qinf$), a functional transducer $\tran$, with $\funcDefinedBy{\tran} \subseteq \alphI^\omega \times \alphO^\omega$ (resp., $\funcDefinedBy{\tran} \subseteq \alphI^* \times \alphO^*$,), is called $K$-Lipschitz robust w.r.t. $\dI$, $\dO$ if:
\[\forall \s,\str \in \dom{\tran}: \ \dI(\s,\str) < \infty \, \Rightarrow \, 
\dO(\funcDefinedBy{\tran}(\s),\funcDefinedBy{\tran}(\str)) \le K \dI(\s,\str).
\]
\end{definition}


Recall that when $\tran$ is an arbitrary nondeterministic transducer,
for each $\s \in \dom{\tran}$, $\funcDefinedBy{\tran}(\s)$
is a set of words in
$\alphO^\omega$ (resp., $\alphO^*$). Hence, we cannot use a similarity function
over $\alphO^\omega$ (resp., $\alphO^*)$ to define the similarity between
$\funcDefinedBy{\tran}(\s)$ and $\funcDefinedBy{\tran}(\str)$, for $\s,\str \in
\dom{\tran}$. Instead, we must use a {\em set-similarity function} that can
compute the similarity between sets of words in $\alphO^\omega$ (resp.,
$\alphO^*$). We define $K$-Lipschitz robustness of nondeterministic transducers
using such set-similarity functions (we use the notation $d$ and $D$ for
similarity functions and set-similarity functions, respectively).

\begin{definition}[$K$-Lipschitz Robustness of Nondeterministic Transducers] 
\label{def:robust-nondet}
Given a constant $K \in \Q$ with $K > 0$, a similarity function 
$\dI:\alphI^\omega \times \alphI^\omega \; \to \; \Qinf$ (resp., $\dI:\alphI^* \times \alphI^* \; \to \; \Qinf$) and a set-similarity function 
$\dsetO: 2^{\alphO^\omega} \times 2^{\alphO^\omega} \; \to \; \Qinf$ (resp., $\dsetO: 2^{\alphO^*} \times 2^{\alphO^*} \; \to \; \Qinf$), a nondeterministic 
transducer $\tran$, with $\funcDefinedBy{\tran} \subseteq \alphI^\omega \times \alphO^\omega$ (resp. $\funcDefinedBy{\tran} \subseteq \alphI^* \times \alphO^*$), 
is called $K$-Lipschitz robust w.r.t. $\dI$, $\dsetO$ if:
\[\forall \s,\str \in \dom{\tran}: \ \dI(\s,\str) < \infty \, \Rightarrow \, 
\dsetO(\funcDefinedBy{\tran}(\s),\funcDefinedBy{\tran}(\str)) \le K \dI(\s,\str).
\]

\end{definition}

In what follows, we use $K$-robustness to denote $K$-Lipschitz robustness. 
The results in the remainder of this paper hold both for 
machines processing $\omega$-words as well as for those processing finite words. 
To keep the presentation clean, we present all results in the context of 
machines over $\omega$-words, making a distinction as needed.

\section{Synchronized (Functional) Transducers}\label{sec:regular}
In this section, we define a class of functional transducers which admits 
a decision procedure for $K$-robustness. 

\begin{definition}[Synchronized Transducers]
\label{def:regular}
A functional transducer $\tran$ with
$\funcDefinedBy{\tran} \subseteq \alphI^\omega \times \alphO^\omega$ 
is synchronized iff there exists an automaton $\aut_\tran$ over $\alphI \otimes \alphO$ 
recognizing the language $\{\s \otimes \funcDefinedBy{\tran}(\s): \s \in \dom{\tran}\}$.  
\end{definition}

Let $\tran$ be an arbitrary functional transducer. In each transition, $\tran$
reads a single input letter and may generate an empty output word or an output word
longer than a single letter. To process such non-aligned input and output words, 
the automaton $\aut_\tran$ needs to internally implement a buffer. 
Thus, $\tran$ is synchronized iff there is a bound $B$ on the required 
size of such a buffer. We can use this observation to check if $\tran$ is synchronized. 
Note that letter-to-letter transducers are synchronised, with $B$ being $0$.

\begin{restatable}{proposition}{RegularityIsDecidable}
\label{prop:RegularityIsDecidable}
Synchronicity of a functional transducer is decidable in polynomial time.
\end{restatable}

%
%

Synchronized transducers admit an automata-theoretic decision procedure 
for checking $K$-robustness w.r.t. similarity functions satisfying certain properties. 

\begin{restatable}{theorem}{RobustnessIsDecidable}
\label{l:simple-regular} 
\label{thm:decidable.regular} 
For every $f \in \valFun$, if $\dI$, $\dO$ are similarity functions
computed by functional $f$-\wa $\aut_{\dI}$, $\aut_{\dO}$, respectively, and $\tran$ is a synchronized transducer, 
$K$-robustness of $\tran$ w.r.t. $\dI, \dO$ is decidable in
polynomial time in the sizes of $\tran$, $\aut_{\dI}$ and $\aut_{\dO}$.
\end{restatable}

We show that for every $f \in \valFun$,
if the conditions of \thmref{decidable.regular} are met, 
$K$-robustness of $\tran$ can be reduced to the emptiness problem for
$f$-weighted automata, which is decidable in polynomial time.\\

\noindent {\em Similarity functions computed by nondeterministic automata.}
If we permit the weighted automata computing the similarity functions $\dI, \dO$ to be nondeterministic, $K$-robustness becomes undecidable. 
We can show that the universality problem for nondeterministic weighted automata  
reduces to checking $1$-robustness.
Indeed, given a nondeterministic weighted automaton $\aut$, consider (1)~$\dI$ such that $\forall \s,\str \in \alphI^{\omega}$: 
$\dI(\s,\str) = \lambda$ if $\s = \str$, and undefined otherwise,  (2)~$\tran$ encoding the 
identity function, and (3)~$\dO$ such that $\forall \s',\str' \in \alphI^{\omega}$: 
$\dO(\s',\str') = \lang_\aut(\s')$ if $\s' = \str'$, and undefined otherwise. 
Note that $\dO$ is computed by a nondeterministic weighted automaton obtained 
from $\aut$ by changing each transition $(q,a,q')$ in $\aut$ to $(q,(a,a),q')$ while 
preserving the weight.  Then, $\tran$ is $1$-robust w.r.t. $\dI, \dO$ 
iff for all words $\s$, $\lang_\aut(\s) \leq \lambda$. Since the universality problem for $f$-weighted automata is undecidable (e.g., for $f = \fsum$), it follows that checking $1$-robustness of transducers 
with similarity functions computed by nondeterministic weighted automata is undecidable. 


We now present examples of synchronized transducers and automatic similarity
functions satisfying the conditions of \thmref{decidable.regular}.
\begin{example}{\bf Mealy machines and generalized Manhattan distances.}  Mealy
machines are perhaps the most widely used transducer model. Prior work
\cite{foundation} has shown decidability of robustness of Mealy machines with
respect to generalized Manhattan distances given a fixed bound on the amount of
input perturbation. In what follows, we argue the decidability of robustness of
Mealy machines (processing infinite words) with respect to generalized
Manhattan distances in the presence of unbounded input perturbation.

A Mealy machine $\tran: (\alphI,\alphO,Q,\{q_0\},E_\tran,Q)$ is a synchronized
transducer with $\aut_\tran$ given by $(\alphI \otimes \alphO, Q,
\{q_0\},E_{\aut_\tran},Q)$, where $E_{\aut_\tran} = \{ (q, a \otimes a', q') :
(q,a,a',q') \in E_\tran \}$.  The generalized \lnorm  
$\hd: \alphI^\omega \times \alphI^\omega \to \Qinf$ can be computed by a 
functional $\fsum$-weighted automaton $\aut_{M}$
given by the tuple $(\alphI \otimes \alphI,\{q_0\}, \{q_0\}, E_M, \{ q_0 \},
\cost)$.  Here, $q_0$ is the initial as well as the accepting state, $E_M =
\{(q_0, a \otimes b, q_0): a \otimes b \in \alphI \otimes \alphI 
 \}$, and the weight of each transition $(q_0, a \otimes
b, q_0)$ equals $\diff(a,b)$. 

Thus, all the conditions of \thmref{decidable.regular} are satisfied. 
$K$-robustness of Mealy machines, when $\dI$ and $\dO$ are computed 
as the generalized \lnorm, is decidable in polynomial time. 
\end{example}

\begin{example}{\bf Piecewise-linear functions.} 
\label{ex:picewise-linear}
Let us use $\fraction{q}$ to denote an infinite word over $\{0, \ldots,9, +,-\}$ 
representing the fractional part
of a real number in base 10.  E.g., $\fraction{-0.21} = -21$ and $\fraction{\pi - 3} = 1415\ldots$ 
Then, $\fraction{q_1} \otimes \ldots \otimes \fraction{q_k}$ is a word 
over $\{0, \ldots,9, +,-\} \otimes \ldots \otimes \{0, \ldots,9, +,-\}$ that
represents a $k$-tuple of real numbers $q_1,\ldots, q_k$ from the interval $(-1,1)$.
Now, observe that one can define letter-to-letter transducers that compute the following functions:
(1)~swapping of arguments, $\funcDefinedBy{\trans}(\fraction{q_1}, \ldots, \fraction{q_l}, \ldots, \fraction{q_m}, \ldots , \fraction{q_k}) = 
(\fraction{q_1}, \ldots, \fraction{q_m}, \ldots, \fraction{q_l}, \ldots , \fraction{q_k})$,
(2)~addition, $\funcDefinedBy{\trans}(\fraction{q_1}, \ldots, \fraction{q_k}) = (\fraction{q_1+q_2}, \fraction{q_2},  \ldots, \fraction{q_k})$,
(3)~multiplication by a constant $c$, $\funcDefinedBy{\trans}(\fraction{q_1}, \ldots, \fraction{q_k}) = 
(\fraction{c q_1}, \ldots, \fraction{c q_k})$,
(4)~projection, $\funcDefinedBy{\trans}(\fraction{q_1}, \ldots, \fraction{q_k}) = (\fraction{q_1}, \ldots, \fraction{q_{k-1}})$, and 
(5)~conditional expression, $\funcDefinedBy{\trans}(\fraction{q_1}, \ldots, \fraction{q_k})$ equals  
$\funcDefinedBy{\trans_1}(\fraction{q_1}, \ldots, \fraction{q_k})$, 
if $\fraction{q_1} > 0$, and $\funcDefinedBy{\trans_2}(\fraction{q_1}, \ldots, \fraction{q_k})$ otherwise. 
We assume that the transducers reject if the results of the corresponding functions 
lie outside the interval $(-1,1)$. We can model a large class of piecewise-linear functions using transducers obtained by 
composition of transducers (1)-(5).  The resulting transducers are functional letter-to-letter transducers.

Now, consider $\dI$, $\dO$ defined as the $\mathit{L1}$-norm over $\mathbb{R}^k$, i.e.,
$\dI(\fraction{q_1} \otimes \ldots \otimes \fraction{q_k},
\fraction{q_1'} \otimes \ldots \otimes \fraction{q_k'})$ $=$ 
 $\dO(\fraction{q_1} \otimes \ldots \otimes \fraction{q_k},
\fraction{q_1'} \otimes \ldots \otimes \fraction{q_k'})$
$=$ $\sum_{i=1}^k \text{abs}(q_i - q_i')$. 
Observe that $\dI, \dO$ can be computed by
deterministic $\fdisc$-weighted automata, with $\delta =
\frac{1}{10}$. Therefore,  
$1$-robustness of $\tran$ can be decided in polynomial
time (\thmref{decidable.regular}). Finally, note that $K$-robustness of a transducer computing a  
piecewise-linear function $h$ w.r.t. the above similarity functions is equivalent 
to Lipschitz continuity of $h$ with coefficient $K$. 
\end{example}

\section{Functional Transducers}\label{sec:functional}
It was shown in \cite{foundation} that checking $K$-robustness of a functional
transducer w.r.t. to a fixed bound on the amount of input perturbation is
decidable.  In what follows, we show that when the amount of input
perturbation is unbounded, the robustness problem becomes undecidable 
even for deterministic transducers.

\begin{restatable}{theorem}{FunctionalUndecidable}
$1$-robustness of deterministic transducers is undecidable.
\label{th:functional-undecidable}
\label{thm:funcundec}
\end{restatable}

\begin{proof} 
The Post Correspondence Problem (PCP) is defined as follows. 
Given a set of word pairs $\{\lpair{v_1}{w_1}, \dots, \lpair{v_k}{w_k}\}$,
does there exist a sequence of indices $i_1, \dots, i_n$ such that
$v_{i_1}\cdot\ldots\cdot v_{i_n} = w_{i_1}\cdot \ldots \cdot w_{i_n}$? 
PCP is known to be undecidable.
 
Let $\PCPinst_{pre} = \{ \lpair{v_1}{w_1}, \dots, \lpair{v_{k}}{w_{k}} \}$ be a
PCP instance with $v_i, w_i \in \{a,b \}^*$ for each $i \in [1,k]$. We define a
new instance $\PCPinst = \PCPinst_{pre} \cup \{\lpair{v_{k+1}}{w_{k+1}}\}$,
where $\lpair{v_{k+1}}{w_{k+1}} = \lpair{\$}{\$}$.  Observe that 
for $i_1, \dots, i_n \in [1,k]$,
$i_1, \dots, i_n, k+1$  is a solution of $\PCPinst$ iff 
$i_1, \dots, i_n$ is a solution of
$\PCPinst_{pre}$. We define a deterministic transducer $\trans$ processing finite words 
and generalized
Manhattan distances $\dI, \dO$ such that $\trans$ is \emph{not} $1$-robust
w.r.t. $\dI, \dO$ iff $\PCPinst$ has a solution of the form $i_1, \dots, i_n,
k+1$, with $i_1, \dots, i_n \in [1,k]$.

We first define $\trans$, which translates indices into corresponding words from the PCP
instance $\PCPinst$. The input alphabet $\alphI$ is the set of indices from
$\PCPinst$, marked with a \emph{polarity}, $L$ or $R$, denoting whether an
index $i$, corresponding to a pair $\lpair{v_i}{w_i} \in \PCPinst$, 
is translated to $v_i$ or $w_i$.
Thus, 
$\alphI = \{1,\ldots, k+1\} \times \{ L,R\}$. 
The output alphabet $\Gamma$ is the alphabet of words in $\PCPinst$, marked with
a polarity. Thus, 
$\Gamma = \{a,b,\$\} \times \{ L,R\}$.  
The domain of $\funcDefinedBy{\trans}$ is described by the following
regular expression: 
$\dom{\trans} =   \Sigma_L^* \lpair{k+1}{L} +  \Sigma_R^*\lpair{k+1}{R}$,
where for $P \in \{L,R\}$, $\Sigma_P = \{1,\dots, k\} \times \{P\}$.
Thus, $\tran$ only processes input words over letters
with the same polarity, rejecting upon reading an input letter with a polarity
different from that of the first input letter. Moreover, $\tran$ accepts iff
the first occurrence of $\lpair{k+1}{L}$ or $\lpair{k+1}{R}$ is in the last
position of the input word.  Note that the domain of $\trans$ is
\emph{prefix-free}, i.e., if $\s,\str \in \dom{\trans}$ and $\s$ is a prefix of
$\str$, then $\s = \str$. Let $u^P$ denote the word $u \otimes P^{|u|}$.
Along accepting runs, $\tran$ translates
each input letter $\lpair{i}{L}$ to $v_i^L$ and each letter $\lpair{i}{R}$ to
$w_i^R$, where $\lpair{v_i}{w_i}$ is the $i^{th}$ word pair of $\PCPinst$.
Thus, the function computed by $\tran$ is: 
\begin{align*}
\funcDefinedBy{\trans}(\lpair{i_1}{L} \dots \lpair{i_n}{L}\lpair{k+1}{L}) &=
v_{i_1}^L \dots v_{i_n}^L v_{k+1}^L \\ \funcDefinedBy{\trans}(\lpair{i_1}{R}
\dots \lpair{i_n}{R}\lpair{k+1}{R}) &= w_{i_1}^R \dots w_{i_n}^R w_{k+1}^R
\end{align*}

We define the output similarity function $\dO$ as a generalized Manhattan distance with
the following {\em symmetric} $\diffO$ where $P,Q \in \{ L,R\}$ and $\alpha, \beta
\in \{a,b, \$\}$ with $\alpha \neq \beta$:
\begin{center}
\begin{tabular}{l l}
$\diffO(\lpair{\alpha}{P},\lpair{\alpha}{P}) = 0$ &  $\;\;\;\diffO(\lpair{\alpha}{L},\lpair{\alpha}{R}) = 2$ \\ 
	$\diffO(\lpair{\alpha}{P},\lpair{\beta}{Q}) = 1$ & $\;\;\;\diffO(\lpair{\alpha}{P},\eos) = 1$\\
\end{tabular}
\end{center}
\noindent Note that for $\s',\str' \in \Gamma^*$ with different polarities, 
$\dO(\s',\str')$ equals the sum of $max(|\s'|, |\str'|)$ and $\N(\s',\str')$, 
where $\N(\s',\str')$ is the number of positions in which 
$\s'$ and $\str'$ agree on the first components of their letters. 

Let us define a projection $\pi$ as
$\pi(\lpair{i_1}{P_1} \lpair{i_2}{P_2} \dots \lpair{i_n}{P_n}) = i_1 i_2 \dots i_n$,
where $i_1,\dots,  i_n \in [1,k+1]$ and $P_1, \dots, P_n \in \{ L, R\}$.
We define the input similarity function $\dI$ as a generalized Manhattan distance such that 
$\dI(\s, \str)$ is finite iff $\pi(\s)$ is a prefix of $\pi(\str)$ or vice versa. 
We define $\dI$ using the following {\em symmetric} $\diffI$ where 
$P,Q \in \{ L,R\}$ and $i,j \in [1,k+1]$ with $i \neq j$:
\begin{center}
\begin{tabular}{l l}
$\diffI(\lpair{i}{P},\lpair{i}{P}) = 0$ & 
$\;\;\; \diffI(\lpair{i}{P},\lpair{j}{Q}) = \infty$ \\
$\diffI(\lpair{i}{L},\lpair{i}{R}) = |v_{i}|+ |w_i|,  \text{ if } i \in [1,k]$ & 
$\;\;\;  \diffI(\lpair{i}{P}, \eos) = \infty$ \\
$\diffI(\lpair{k+1}{L},\lpair{k+1}{R}) = 1$ & \\ 
\end{tabular}
\end{center}
Thus, for all $\s,\str \in \dom{\trans}$, $\dI(\s,\str) < \infty$ iff one of the following holds:
\begin{compactenum}[(i)]
\item for some $P \in \{L,R\}$, $\s = \str = \lpair{i_1}{P} \dots \lpair{i_n}{P} \lpair{k+1}{P}$, or,  
\item $\s = \lpair{i_1}{L} \dots \lpair{i_n}{L}\lpair{k+1}{L} $ and 
$\str = \lpair{i_1}{R} \dots \lpair{i_n}{R} \lpair{k+1}{R}$. 
\end{compactenum}
\medskip
In case (i), $\dI(\s,\str) = \dO(\funcDefinedBy{\trans}(\s),\funcDefinedBy{\trans}(\str)) = 0$. 
In case (ii), $\dI(\s, \str) = |\funcDefinedBy{\trans}(\s)| + |\funcDefinedBy{\trans}(\str)|-1$ 
and $\dO(\funcDefinedBy{\trans}(\s),\funcDefinedBy{\trans}(\str)) = 
max(|\funcDefinedBy{\trans}(\s)|, |\funcDefinedBy{\trans}(\str)|) +  
\N(\funcDefinedBy{\trans}(\s), \funcDefinedBy{\trans}(\str))$. 
Thus, $\dO(\funcDefinedBy{\trans}(\s), \funcDefinedBy{\trans}(\str)) > \dI(\s,\str)$ iff  
$\N(\funcDefinedBy{\trans}(\s), \funcDefinedBy{\trans}(\str)) = min(|\funcDefinedBy{\trans}(\s)|, |\funcDefinedBy{\trans}(\str)|)$. 
Since the letters $\lpair{\$}{L}, \lpair{\$}{R}$ occur exactly once in $\funcDefinedBy{\trans}(\s)$, $\funcDefinedBy{\trans}(\str)$,
respectively, at the end of each word, $\N(\funcDefinedBy{\trans}(\s), \funcDefinedBy{\trans}(\str)) = min(|\funcDefinedBy{\trans}(\s)|, |\funcDefinedBy{\trans}(\str)|)$ iff $|\funcDefinedBy{\trans}(\s)| = |\funcDefinedBy{\trans}(\str)|$ and 
$\pi(\funcDefinedBy{\trans}(\s)) = \pi(\funcDefinedBy{\trans}(\str))$, which holds 
iff $\PCPinst$ has a solution. 
Therefore, $\trans$ is \emph{not} $1$-robust w.r.t. $\dI, \dO$ iff
$\PCPinst$ has a solution.
\end{proof}


We have shown that checking $1$-robustness w.r.t. generalized Manhattan
distances is undecidable. Observe that for every $K>0$, $K$-robustness can be 
reduced to $1$-robustness by scaling the output distance by $K$.  We conclude that 
checking $K$-robustness is undecidable for any fixed $K$.  In contrast, if $K$ is
not fixed, checking if there exists $K$ such that $\tran$ is $K$-robust w.r.t. $\dI$, $\dO$ 
is decidable for transducers processing finite words.

Let us define a functional transducer $\tran$ to be \emph{robust} w.r.t. $\dI$,
$\dO$ if there exists $K$ such that $\tran$ is $K$-robust w.r.t. $\dI$, $\dO$.


\begin{restatable}{proposition}{IfKrobustThenRobust}
Let $\tran$ be a given functional transducer processing finite words and $\dI$, $\dO$ be 
instances of the generalized Manhattan distance. 
\begin{enumerate}
\item Robustness of $\tran$ is decidable in \conp. 
\item One can compute $K_{\tran}$ such that $\tran$ is robust  
iff $\tran$ is $K_{\tran}$-robust.
\end{enumerate}
\end{restatable}

\vspace{0.05in}
\noindent {\em Proof sketch}. 
Given $\tran$, one can easily construct a {\em trim}\footnote{$\ptran$ is trim
if every state in $\ptran$ is reachable from the initial state and some final
state is reachable from every state in $\ptran$.} functional transducer
$\ptran$ such that $\funcDefinedBy{\ptran}(\s,\str) = (\s',\str')$ iff
$\funcDefinedBy{\tran}(\s) = \s'$ and $\funcDefinedBy{\tran}(\str) = \str'$. We
show that $\tran$ is not robust w.r.t. generalized Manhattan distances iff there
exists some cycle in $\ptran$ satisfying certain properties.  Checking 
the existence of such a cycle is in \np. If such a cycle exists, one
can construct paths in $\ptran$ through the cycle, labeled with input words
$(\s, \str)$ and output words $(\s',\str')$, with ${\dO}(\s',\str') >
K{\dI}(\s,\str)$ for {\em any} $K$.
Conversely, if there exists no such cycle,
one can compute $K_{\tran}$ such that $\tran$ is $K_{\tran}$-robust. It follows
that one can compute $K_{\tran}$ such that $\tran$ is robust iff $\tran$ is
$K_{\tran}$-robust.

%
%

\subsection{Beyond Synchronized Transducers}

In this section, we present an approach for natural extensions of 
\thmref{decidable.regular}.\\ 


\noindent {\bf Isometry approach.} We say that a transducer $\trans$ is a \emph{$(\dIi, \dOo)$-isometry} if and only if
for all $\s, \str \in \dom{\trans}$ we have
$\dIi(\s,\str) = \dOo(\funcDefinedBy{\trans}(\s),\funcDefinedBy{\trans}(\str))$.

\begin{proposition}
\label{prop:compositions}
Let $\trans, \trans'$ be functional transducers 
with $\funcDefinedBy{\trans} \subseteq \alphI^\omega \times \alphO^\omega$ and
$\funcDefinedBy{\trans'} \subseteq \alphIi^\omega \times \alphOo^\omega$.
Assume that there exist transducers $\trans^I$ and $\trans^O$ such that 
$\trans^I$ is a $(\dI, \dIi)$-isometry, $\trans^O$  is a $(\dOo, \dO)$-isometry and 
$\funcDefinedBy{\trans} = \funcDefinedBy{\trans^O \circ (\trans' \circ \trans^I)}$.
Then, for every $K>0$,  $\trans$ is $K$-robust w.r.t. $\dI, \dO$ if and only if
$\trans'$ is $K$-robust w.r.t. $\dIi, \dOo$.
\end{proposition} 

\begin{example}[Stuttering]
\newcommand{\transPi}{\trans^{\pi}}
\newcommand{\transS}{\trans^{S}}
\newcommand{\transI}{\trans^{I}}
For a given word $w$ we define the \emph{stuttering pruned} word $\stutter(w)$ as the result 
of removing from $w$ letters that are the same as the previous letter.
E.g. $\stutter(\mathbf{ba}aa\mathbf{c}c\mathbf{a}aa\mathbf{b}) = bacab$. 

Consider a transducer $\trans$ and a similarity function $\dI$ over finite words that are \emph{stuttering invariant}, \ie, for all $\s,\str \in \dom{\trans}$, if $\stutter(\s) = \stutter(\str)$, then
$\funcDefinedBy{\trans}(\s) = \funcDefinedBy{\trans}(\str)$
and for every $\sStr \in \alphI^*$, $\dI(\s,\sStr) = \dI(\str,\sStr)$.
In addition, we assume that for every $\s \in \dom{\trans}$,
$|\funcDefinedBy{\trans}(\s)| = |\stutter(\s)|$.

Observe that these assumptions imply that:
(1)~the projection transducer $\transPi$ defined such that
$\funcDefinedBy{\s} = \stutter(\s)$ is a $(\dI, \dI)$-isometry,
(2)~the transducer $\transS$ obtained by restricting the domain of $\trans$ 
to \emph{stuttering-free} words, \ie, the set $\{ w \in \dom{\trans} : \stutter(w) = w \}$,
is a synchronized transducer\footnote{Note that any functional transducer $\trans$ with the property:   
for every $\s \in \dom{\trans}$, $|\funcDefinedBy{\trans}(\s)| = |\s|$, is a synchronized transducer.}, 
 and 
(3)~$\funcDefinedBy{\trans} = \funcDefinedBy{\transI \circ (\transS \circ \transPi)}$, where
$\transI$ defines the identity function over $\alphO^*$.
Therefore, by \propref{compositions}, in order to check $K$-robustness of $\trans$, it suffices 
to check $K$-robustness of $\transS$. Since $\transS$ is a synchronized transducer, 
$K$-robustness of $\transS$ can be effectively checked,
provided the similarity functions $\dI, \dO$ satisfy the conditions of \thmref{decidable.regular}.
\end{example}

\begin{example}[Letter-to-multiple-letters transducers]
\newcommand{\transPair}{\trans^{\textrm{pair}}}
\newcommand{\transD}{\trans^{D}}
\newcommand{\transI}{\trans^{I}}
Consider a transducer $\trans$   
which on every transition outputs a $2$-letter 
word\footnote{One can easily generalize this example to any fixed number.}.
Although, $\trans$ is not synchronized, it can be transformed to 
a letter-to-letter transducer $\transD$, whose output alphabet is $\alphO \times \alphO$.
The transducer $\transD$ is obtained from $\trans$ by substituting each output word $ab$ to a single letter $\lpair{a}{b}$
from $\alphO \times \alphO$. 
We can use $\transD$ to decide $K$-robustness of $\trans$ in the following way.
First, we define transducers $\transI, \transPair$ such that
$\transI$ computes the identity function over $\alphI^\omega$ and
$\transPair$ is a transducer representing the function 
$\funcDefinedBy{\transPair}(\lpair{a_1}{b_1}\lpair{a_2}{b_2}\ldots) = a_1 b_1 a_2 b_2 \ldots$.
Observe that $\funcDefinedBy{\trans} = \funcDefinedBy{\transPair \circ (\transD \circ \transI)}$.
Second, we define $\dO^D$ as follows: $\forall \s,\str \in (\alphI\times\alphI)^\omega$,
$\dO^D(\s,\str) = \dO(\funcDefinedBy{\transPair}(\s), \funcDefinedBy{\transPair}(\str))$.
Observe that  $\transI$ is a $(\dI,\dI)$-isometry and 
$\transPair$ is a $(\dO^D,\dO)$-isometry.
Thus, $K$-robustness of $\trans$ w.r.t. $\dI, \dO$ reduces to
$K$-robustness of the letter-to-letter transducer $\transD$ w.r.t.
$\dI, \dO^D$, which can be effectively checked (\thmref{decidable.regular}).
\end{example}


\section{Nondeterministic Transducers}\label{sec:nondeter}
Let $\tran$ be a nondeterministic transducer with $\funcDefinedBy{\tran}
\subseteq \alphI^\omega \times \alphO^\omega$. Let $\dI$ be an automatic
similarity function for computing the similarity between input words in
$\alphI^*$. As explained in \secref{probdef}, the definition of $K$-robust
nondeterministic transducers 
involves set-similarity functions that
can compute the similarity between sets of output words in $\alphO^\omega$. In
this section, we examine the $K$-robustness problem of $\tran$
w.r.t. $\dI$ and three classes of such set-similarity functions. 

Let $\dO$ be an automatic similarity function for computing the similarity
between output words in $\alphO^\omega$. We first define three
set-similarity functions induced by $\dO$.
\begin{definition}
Given sets $A, B$ of words in $\alphO^\omega$, we consider the following set-similarity functions 
induced by $\dO$:
\begin{enumerate}[(i)]
\item Hausdorff set-similarity function $\dsetO^{H}(A,B)$ induced by $\dO$:
\label{def:dHaus}
\[
\dsetO^{H}(A,B) = max \{\, \sup\nolimits_{\s \in A} \, \inf\nolimits_{\str \in B} \dO(\s,\str), 
			   \sup\nolimits_{ \s \in B} \, \inf\nolimits_{\str \in A} \dO(\s,\str) \, \}
\]
\item Inf-inf set-similarity function $\dsetO^{\inf}(A,B)$ induced by $\dO$:
\label{def:dinf}
\[
\dsetO^{\inf}(A,B) = \inf\nolimits_{\s \in A} \inf\nolimits_{\str \in B} \dO(\s,\str)
\]
\item Sup-sup set-similarity function $\dsetO^{\sup}(A,B)$ induced by $\dO$:
\label{def:dsup}
\[
\dsetO^{\sup}(A,B) = \sup\nolimits_{\s \in A} \sup\nolimits_{\str \in B} \dO(\s,\str)
\]
\end{enumerate}
\end{definition}

\noindent Of the above set-similarity functions, only the Hausdorff
set-similarity function is a distance function (if $\dO$ is a distance function). 

Note that when $\tran$ is a functional transducer, each 
set-similarity function above reduces to $\dO$.  Hence, $K$-robustness of a
functional transducer $\tran$ w.r.t. $\dI$, $\dsetO$ and $K$-robustness of
$\tran$ w.r.t. $\dI$, $\dO$ coincide.  As $K$-robustness of functional
transducers in undecidable (\thmref{funcundec}), $K$-robustness of
nondeterministic transducers w.r.t. the above set-similarity functions is
undecidable as well.

%
Recall from \thmref{decidable.regular} that $K$-robustness of a
synchronized (functional) transducer is decidable w.r.t. certain automatic
similarity functions.  In particular, $K$-robustness of Mealy machines
is decidable when $\dI$, $\dO$ are generalized Manhattan distances.  
In contrast, $K$-robustness of nondeterministic letter-to-letter
transducers is undecidable w.r.t. the Hausdorff 
and Inf-inf set-similarity functions even
when $\dI$, $\dO$ are generalized Manhattan distances.  
Among the above defined set-similarity functions, $K$-robustness of
nondeterministic transducers is decidable only w.r.t. the Sup-sup
set-similarity function.

\begin{restatable}{theorem}{RobustnessdNon}
\label{thm:HausUndecidable}
Let $\dI$, $\dO$ be computed by functional weighted-automata. 
Checking $K$-robustness of nondeterministic letter-to-letter transducers w.r.t. $\dI$, 
$\dsetO$ induced by $\dO$ is  
\begin{compactenum}[(i)]
\item undecidable if $\dsetO$ is the Hausdorff set-similarity function,
\item undecidable if $\dsetO$ is the Inf-inf set-similarity function, and
\item decidable  if $\dsetO$ is the Sup-sup set-similarity function and $\dI, \dO$ 
satisfy the conditions of \thmref{decidable.regular}.
\end{compactenum}
\end{restatable}

\begin{proof} {\bf [of (iii)]}
We can encode nondeterministic choices of $\trans$, with $\funcDefinedBy{\tran} \subseteq \alphI^\omega \times \alphO^\omega$,
in an extended input alphabet $\alphI \times \Lambda$. We construct a deterministic transducer $\trans^e$ such that 
for every $\s \in \alphI^\omega$, 
$\{ \funcDefinedBy{\trans^{e}}(\lpair{\s}{\lambda}) : \lpair{\s}{\lambda} \in \dom{\trans^{e}} \} = 
\funcDefinedBy{\trans}(\s)$.
We  also define $\dI^e$ such that for all $\lpair{\s}{\lambda_1}, \lpair{\str}{\lambda_2} \in (\alphI \times \Lambda)^\omega$,
$\dI^e(\lpair{\s}{\lambda_1}, \lpair{\str}{\lambda_2}) = \dI(\s,\str)$.
Then, $\trans$ is $K$-robust w.r.t. $\dI, \dsetO^{\sup}$ induced by $\dO$ 
iff $\trans^e$ is $K$-robust w.r.t. $\dI^e, \dO$.
Indeed, a nondeterministic transducer $\trans$ is $K$-robust
w.r.t. $\dI$, $\dsetO^{\sup}$ induced by $\dO$ 
iff for all input words $\s,\str \in \dom{\trans}$ and for
all outputs $\s' \in \funcDefinedBy{\trans}(\s), \str' \in \funcDefinedBy{\trans}(\str)$,
$\dI(\s,\str) < \infty$ implies $\dO(\s',\str') \leq K \dI(\s,\str)$. 
\end{proof}

\section{Related Work}\label{sec:related}

In early work  \cite{RTSS09}, \cite{CGL10,CGLN11} on continuity and robustness
analysis, the focus is on software programs manipulating numbers.  In
\cite{RTSS09}, the authors compute the maximum deviation of a program's output
given the maximum possible perturbation in a program input. In \cite{CGL10},
the authors formalize $\epsilon-\delta$ continuity of programs and present
sound proof rules to prove continuity of programs. In \cite{CGLN11}, the
authors formalize robustness of programs as Lipschitz continuity and present a
sound program analysis for robustness verification. While arrays of numbers are
considered in \cite{CGLN11}, the size of an array is immutable. 
 
More recent papers have aimed to develop a notion of robustness for
reactive systems. In \cite{TBCSM12}, the authors present
polynomial-time algorithms for the analysis and synthesis of robust
transducers. Their notion of robustness is one of input-output
stability, that bounds the output deviation from
disturbance-free behaviour under bounded disturbance, as well as the
persistence of the effect of a sporadic disturbance.
Their distances are measured using cost functions that 
map {\em each} string to a nonnegative integer. 
In \cite{MRT13,CHR10,BGHJ09}, the authors develop
different notions of robustness for reactive systems, with
$\omega$-regular specifications, interacting with uncertain
environments. In \cite{DHLN10}, the authors present a polynomial-time
algorithm to decide robustness of sequential circuits modeled as Mealy
machines, w.r.t. a {\em common suffix distance} metric.  Their notion
of robustness also bounds the persistence of the effect of a 
sporadic disturbance.

Recent work in \cite{SDC13} and \cite{foundation} formalized and studied
robustness of systems modeled using transducers, in the presence of bounded
perturbation. The work in \cite{SDC13} focussed on the outputs of synchronous
networks of Mealy machines in the presence of channel perturbation. 
The work in \cite{foundation} focussed on the outputs of functional transducers
in the presence of input perturbation. Both papers presented decision 
procedures for robustness verification w.r.t. specific distance functions
such as Manhattan and Levenshtein distances.

\section{Conclusion}\label{sec:conclusion}
In this paper, we studied the $K$-Lipschitz
robustness problem for finite-state transducers.
While the general problem is undecidable,
we identified decidability criteria that enable
reduction of $K$-robustness to the emptiness problem 
for weighted automata. 

In the future, we wish to extend our work in two directions. 
We plan to study robustness of
other computational models. We also wish to investigate 
synthesis of robust transducers. 

\bibliography{papers}

\newpage
\appendix
\section{Proofs}\label{app:proofs}
\RegularityIsDecidable*

\begin{proof}
We prove the result for transducers processing finite words. The proof 
for transducers processing infinite words is similar, but a bit more technical. 
 
The proof consists of two claims:

\noindent(1)~a functional transducer $\tran$ is synchronized iff
there are $B>0$ and a finite set of words $\Lfin$
such that for every word $\s = \s[1] \s[2] \dots \in \dom{\tran}$
there is an accepting run $(q_0, u_1)(q_1, u_2) \dots (q_{n-1}, u_n) (q_{n}, \epsilon)$ on $\s$
such that for every $i \in \{1,\ldots, |\s|\}$

(a)~ $ |u_1 u_2 \dots u_i| - i \leq B$, and

(b)~ if $i - |u_1 u_2 \dots u_i| > B$, then $u_{i+1} u_{i+1} \dots u_{|\s|} \in \Lfin$.

\noindent(2)~it is decidable in polynomial time whether such $B$ exists.
\smallskip

(1)~$\Leftarrow:$ Given $B$ and a finite language $\Lfin$ we can construct an 
automaton that simulates runs of $\trans$. The automaton implements a buffer of
size $B$ used to align input and output words. Due to assumptions on $\trans$,
the buffer will not overflow with output letters (cond. (a)), and if it overflows with
input letters, the remaining words belongs to a finite language (cond. (b), language $\Lfin$). 
In the latter case, once the overflow is detected, the automaton
can nondeterministically guess a word from $\Lfin$, and check correctness of that guess.

$\Rightarrow:$
Assume towards contradiction that such $B, \Lfin$ do not exist but
$\tran$ is synchronized, \ie, there exists an automaton $\aut$  satisfying:
for all $\s,\str$ we have $\funcDefinedBy{\tran}(\s) = \str$ iff $\s \otimes \str \in \lang_{\aut}$.
First suppose {(a)} is violated.
Then, there is a finite word $\s$ such that $|\funcDefinedBy{\tran}(\s)| - |s| > |\aut|$.
Consider an accepting run of $\aut$ on $\s \otimes \funcDefinedBy{\tran}(\s)$. 
That run can be pumped between positions $|\s|$  and $|\s| + |\aut|$ 
to get an accepting run on $\s \otimes u'$, where
$|u'| > |\funcDefinedBy{\tran}(\s)|$, which contradicts functionality of $\trans$.

Suppose (b) is violated.  Consider a state $\hat{q}$ of $\trans$ such that
there exist $k = |\aut| + 1$ inputs $\str_1, \ldots, \str_k$ corresponding to
$k$ pairwise different outputs $v_1, \ldots, v_k$, \ie, the transducer starting
in the state $\hat{q}$ upon reading $\str_i$, produces the word $v_i$.  Let $M
= max(|v_1|, \dots, |v_k|)$.  Suppose that there are $\s = \s[1] \dots \s[i]$
and a run $(q_0, u_1) \dots (q_{i-2}, u_{i-1}) (\hat{q}, u_i)$ such that  $|\s|
- |u_1 u_2 \dots u_i| > M$.  We have $\funcDefinedBy{\trans}(\s \str_1) = u
  v_1, \ldots, \funcDefinedBy{\trans}(\s \str_k) = u v_k$, where $u = u_1 u_2
\dots u_i$. Thus, $\aut$ accepts $\s \str_1 \otimes u v_1, \dots, \s \str_k
\otimes u v_k$.  As $k > |\aut|$ there are two different indices $i,j$ such
that $\aut$ is in the same state after reading $|\s|$-letter prefixes of $\s
\str_i \otimes u v_i$ and $\s \str_j \otimes u v_j$.  Since $|v_i|, |v_j| <
M$, $|u v_i|, |u v_j| < |\s|$.  Therefore, $\aut$ accepts $\s \str_i \otimes
u v_j$, which contradicts functionality of $\trans$.  \smallskip

(2)~Consider a functional transducer $\tran$. Let $Q'$ be a subset of states
of $Q$ consisting of states that are reachable from the initial state
and  from which some final states are reachable. The condition 
(a) holds if and only if there are no cycles in which the output word is longer than the input word.
The condition (b) holds if and only if all states reachable from cycles with empty output word
have output languages finite, \ie, once transducer enters such a state, there is a bounded number
of possible words it can output.
\end{proof}

\RobustnessIsDecidable*

\begin{proof}
Let $\aut_{\dI}$ 
be an $f$-weighted automaton computing the similarity function $\dI$.  Let $\aut_{\dO}$
be a functional $f$-weighted automaton
computing the similarity function $\dO$.  Let $\aut_{\tran}$
be the automaton corresponding to $\tran$ as defined in
\defref{regular}.
We define variants of $\aut_{\dI}, \aut_{\dO}$ and $\aut_{\tran}$ to enable these
automata to operate over a common alphabet $\Lambda = \alphI \otimes \alphI
\otimes \alphO \otimes \alphO$.
We obtain $\baut_{\dI}$ by replacing each
transition $(q, \lpair{a}{b},q')$ in $\aut_{\dI}$ with a set
$\{  (q, \langle a, b, a', b' \rangle ,q') : a',b' \in\alphO\}$ of transitions and setting the weight of each transition
$(q, \langle a, b, a', b' \rangle ,q')$ in this set to the weight of $(q, \lpair{a}{b},q')$.
Thus, the value of a run of $\baut_{\dI}$ on a word $\s
\otimes \str \otimes \s' \otimes \str'$ over $\Lambda$ equals the value of the
corresponding run (over the same sequence of states) of $\aut_{\dI}$ on word
$\s \otimes \str$.  This implies that the value of $\baut_{\dI}$ on
$\s \otimes \str \otimes \s' \otimes \str'$ is equal to the value of
$\aut_{\dI}$ on $\s \otimes \str$. In a similar way, we define automata
$\baut_{\dO}, \blaut_{\tran}, \braut_{\tran}$ on $\Lambda$
such that for all $\s,\str,\s',\str'$:
\begin{compactenum}
\item the value of $\baut_{\dO}$ on $s \otimes \str \otimes \s' \otimes \str'$
is equal to the value of $\aut_{\dO}$ on $\s' \otimes \str'$,
\item  $\blaut_{\tran}$ accepts $\s \otimes \str \otimes \s' \otimes \str'$
iff $\aut_{\tran}$ accepts $\s \otimes \s'$, and
\item $\braut_{\tran}$ accepts $\s \otimes \str \otimes \s' \otimes \str'$
iff $\aut_{\tran}$ accepts $\str \otimes \str'$.
\end{compactenum}

Let $\baut_{\dI}^{K}$, $\baut_{\dO}^{-1}$ are $f$-weighted automata obtained by
multiplying each transition weight of $\baut_{\dI}$, $\baut_{\dO}$ by $K$,
$-1$, respectively.
Consider the $f$-weighted automaton $\aut$ defined as $\baut_{\dI}^K \times
\blaut_{\tran} \times \braut_{\tran} \times \baut_{\dO}^{-1}$, the synchronized product
of automata 
$\baut_{\dI}^K, \blaut_{\tran}, \braut_{\tran}, \baut_{\dO}^{-1}$ where
the weight of each transition is equal the sum of the weights of the corresponding transitions
in $\baut_{\dI}^K$ and $\baut_{\dO}^{-1}$.

Now, we show that there exists a word with the value below $0$ assigned by $\aut$ 
iff $\trans$ is not $K$-robust w.r.t. $\dI, \dO$.
We perform case distinction: (1)~$f \in \{ \fsum, \fdiscF,\fdiscI\}$ and (2)~$f = \flimavg$.

In proofs of cases (1) and (2) we use the following notation.
Given sequences of real numbers $\sigma_1, \sigma_2$ and a real number $c$ 
we denote by 
$\sigma_1 + \sigma_2$ the component-wise sum of sequences $\sigma_1$ and $\sigma_2$, and
by $c \cdot \sigma_1$ the component-wise multiplication of $\sigma_1$ by $c$.

(1):~Let $f \in \{ \fsum, \fdiscF,\fdiscI\}$. 
Observe that for all sequences of real numbers $\sigma_1, \sigma_2$
and every real number $c$ we have 
(a)~$f(\sigma_1) + f(\sigma_2)$ equals $f(\sigma_1 + \sigma_2)$, and
(b)~$f(c \cdot \sigma_1) = c \cdot f(\sigma_1)$.

Consider words $\s,\str, \s',\str'$ such that $\aut$ accepts $\s \otimes \str \otimes \s' \otimes \str'$.
The conditions $(a), (b)$ imply that
$\lang_{\aut}(\s \otimes \str \otimes
\s' \otimes \str')$ equals $K\dI(\s,\str) + \inf_{\pi \in \Acc } -f(\pi)$,
where $\Acc$ is the set of accepting runs of $\aut_{\dO}$
on $\s' \otimes \str'$.
As $\aut_{\dO}$ is functional,
each run in $\Acc$ has the same value $\dO(\s', \str')$.
Thus, $\lang_{\aut}(\s \otimes \str \otimes
\s' \otimes \str')$ equals $K\dI(\s,\str) - \dO(\s', \str')$, and,
$\lang_{\aut}(\s \otimes \str \otimes
\s' \otimes \str') < 0$ implies $\tran$ is \emph{not} $K$-robust.
Conversely, if $\tran$ is not $K$-robust there are words $\s, \str$ such that
$\dO(\funcDefinedBy{\trans}(\s), \funcDefinedBy{\trans}(\str)) >
K\dI(\s,\str)$.  This implies $\aut$ accepts $\s \otimes \str \otimes
\funcDefinedBy{\trans}(\s) \otimes \funcDefinedBy{\trans}(\str)$ and
$\lang_{\aut}(\s \otimes \str \otimes \funcDefinedBy{\trans}(\s) \otimes
\funcDefinedBy{\trans}(\str)) < 0$.
Thus, nonemptiness of $\aut$ and $K$-robustness of $\tran$ w.r.t. $\dI, \dO$ coincide

(2):~Let $f = \flimavg$.
Observe that for sequences of real numbers $\sigma_1, \sigma_2$, 
$\limsup_{k\rightarrow \infty} K \cdot \sigma_1[k]- \sigma_2[k] \geq K \limsup_{k\rightarrow \infty} \sigma_1[k] -
\limsup_{k\rightarrow \infty} \sigma_2[k]$. 
Therefore, for sequences of real numbers $\sigma_1, \sigma_2$, $\flimavg(K \sigma_1 - \sigma_2) \leq K \flimavg(\sigma_1) - \flimavg(\sigma_2)$.
It follows that if $\aut$ assigns to every accepted word the value greater of equal to $0$,
$\trans$ is $K$-robust w.r.t. $\dI, \dO$.
Conversely, assume that $\aut$ has a run of the value below $0$. Then,
it has also a run $\pi$ on some word $\s \otimes \str \otimes \s' \otimes \str'$ 
 of the value below $0$ which is a lasso. 
The partial averages of weights is lasso converge, i.e. ,
$\limsup_{k\rightarrow \infty} \frac{1}{k}
\sum_{i=1}^{k} (\cost(\pi))[i] =
\liminf_{k\rightarrow \infty} \frac{1}{k}
\sum_{i=1}^{k} (\cost(\pi))[i]$.
Hence, $0 > \flimavg(\pi) = \flimavg(\pi_1) - \flimavg (\pi_2)$,
where $\pi_1$ (resp. $\pi_2$) is the projection of the run $\pi$ on states of 
$\baut_{\dI}^{K}$ (resp. $\baut_{\dO}^{-1}$). 
As $\flimavg(\pi_1) \geq K\dI(\s,\str)$
and $\flimavg(\pi_2) \geq \dO(\s',\str')$, we have 
$K\dI(\s,\str) < \dO(\s',\str')$.
It follows that if $\aut$ assigns to some word the value below $0$, $\trans$
is not $K$-robust w.r.t. $\dI, \dO$.
\qed
\end{proof}

\noindent Let $\Pair_\tran$ denote the trim functional transducer obtained from
$\tran$ such that $\funcDefinedBy{\ptran}(\s,\str) = (\s',\str')$ iff
$\funcDefinedBy{\tran}(\s) = \s'$ and $\funcDefinedBy{\tran}(\str) = \str'$.
Let $q_{0_\Pair}$ denote an initial state of $\Pair_\tran$ and $q_{F_\Pair}$
denote a final state of $\Pair_\tran$.  We say two words $\s$, $\str$ are {\em
rotations} of each other if $|\s| = |\str|$ and there exist words $u$ and $v$
such that $\s = u.v$ and $\str = v.u$; $(|u|,|v|)$ are referred to as the {\em
gap parameters}.  Note that words $\s$, $\str$ with $\s =\str$ are also
rotations of each other. 

\begin{lemma}\label{lem:suffnonrob}
Let $d$ be a Manhattan distance on $(\alphI \cup \alphO)^*$
and let $\tran$ be a transducer.
The following conditions are equivalent:
\begin{enumerate}
\item  there exists a cycle 
$c: (q_c,(\inp_c,\inpt_c),(\out_c,\outt_c),q_c)$ in 
$\Pair_\tran$ with $d(\inp_c,\inpt_c) = 0$, satisfying one of the following 
conditions:
\begin{enumerate}
\item $\out_c$, $\outt_c$ are not rotations of each other, or, 

\item $\out_c$, $\outt_c$ are rotations of each other with gap parameters 
($r,s$), and there exists a path of $\Pair$ through $c$:
$(q_{0_\Pair},(\inp_{pre},\inpt_{pre}),
(\out_{pre},\outt_{pre}),q_c)$ $.$ $c$
with $abs(|\outt_{pre}| - |\out_{pre}|) \neq r$ 
and $abs(|\outt_{pre}| - |\out_{pre}|) \neq s$,
\end{enumerate}
\item  $\tran$ is not robust w.r.t. Manhattan distances, and
\item $\tran$ is not $K^{\tran}$-robust w.r.t. Manhattan distances,
where $K^{\tran} = |\Pair_\tran|^2$.
\end{enumerate}
\end{lemma}
\begin{proof}
\noindent{\em (i) $\implies$ (ii)}:
Suppose there exists a cycle $c: (q_c,(\inp_c,\inpt_c),(\out_c,\outt_c),q_c)$
in $\Pair$ with $d(\inp_c,\inpt_c) = 0$ satisfying condition $1$ or $2$ in \lemref{suffnonrob}.
If $c$ satisfies condition $1$, let $\pi$ denote any path 
$\pi_{pre}$ $.$ $c$ in $\Pair$ through $c$, with $\pi_{pre}$ given by 
$(q_{0_\Pair},(\inp_{pre},\inpt_{pre}),(\out_{pre},\outt_{pre}),q_c)$.
If $c$ satisfies condition $2$, let $\pi:$ $\pi_{pre}$ $.$ $c$, with 
$\pi_{pre}:$ $(q_{0_\Pair},(\inp_{pre},\inpt_{pre}),(\out_{pre},\outt_{pre}),q_c)$, 
be a path such that $abs(|\outt_{pre}| - |\out_{pre}|) \neq r$ and
$abs(|\outt_{pre}| - |\out_{pre}|) \neq s$.

Then, given an arbitrary $K \in \Nat$, there exists $n \in \Nat$ such that for
the path $\pi_{pre}$ $.$ $(c)^n$,
$d(\out_{pre}.(\out_c)^n,\outt_{pre}.(\outt_c)^n) >
Kd(\inp_{pre}.(\inp_c)^n,\inpt_{pre}.(\inpt_c)^n)$.  Hence, $\tran$ is not
robust. 

\noindent{\em (ii) $\implies$ (iii) }: Follows directly form the definition of robustness.

\noindent{\em (iii) $\implies$ (i)}:
Assume towards contradiction that no cycle in  $\Pair_\tran$ satisfies (a) or (b) in condition 1.

Let $K_{cyc} = max_c (max (|\out_c|,|\outt_c|))$ be the maximum length of an
output word over all cycles $c: (q_c,(\inp_c,\inpt_c),(\out_c,\outt_c),q_c)$
in $\Pair$ with $d(\inp_c,\inpt_c) \neq 0$ \ie, $d(\inp_c,\inpt_c) \geq 1$. Then, for any such cycle,
$d(\out_c,\outt_c) \leq K_{cyc} \leq K_{cyc}d(\inp_c,\inpt_c)$. 

Let $K_{acyc} = (|Q_\Pair|-1)\ell_{max}$, where $Q_\Pair$ is the set of states
of $\Pair$ and $\ell_{max}$ is the length of the longest output word along
any transition of $\tran$. Then, for any acyclic fullpath $\pi: (q_{0_\Pair},
(\inp_\pi,\inpt_\pi),(\out_\pi,\outt_\pi),q_{F_\Pair})$ in $\Pair$, either
$d(\inp_\pi,\inpt_\pi) = 0$ and consequently, since $\tran$ is functional,
$d(\out_\pi,\outt_\pi) = 0$, or, $d(\inp_\pi,\inpt_\pi) \geq 1$ and
consequently, $d(\out_\pi,\outt_\pi) \leq K_{acyc} \leq
K_{acyc}d(\inp_\pi,\inpt_\pi)$.  Note that the maximum possible distance, $K_{acyc}$,
between the output words along an acyclic fullpath of $\Pair$, is exhibited
along an acyclic path of maximum length $|Q_\Pair|-1$, with the output word
pair along each transition being $\eps$ and a word of length $\ell_{max}$. 

Let $K = max(K_{cyc}, K_{acyc})$. We prove below that $\tran$ is $K$-robust.
Observe that $K_{cyc} \leq |\Pair|$, $K_{acyc} \leq |\Pair|\cdot |\Pair|$.
Thus, $K \leq K^{\tran}$ and if  $\tran$ is $K$-robust, it is $K^{\tran}$
robust.

For an arbitrary fullpath $\pi: (q_{0_\Pair},
(\inp_\pi,\inpt_\pi),(\out_\pi,\outt_\pi),q_{F_\Pair})$ of $\Pair$, 
we will prove by induction on the number of cycles in $\pi$ that
$d(\out_\pi,\outt_\pi) \leq Kd(\inp_\pi,\inpt_\pi)$. 

When the number of cycles is $0$, \ie, when $\pi$ is acyclic, it is obvious
from the definitions of $K_{acyc}$ and $K$ that the induction hypothesis holds.
Suppose the induction hypothesis holds for any fullpath with $N \geq 0$ cycles.
Consider a fullpath $(q_{0_\Pair},(\inp_{pre},\inpt_{pre}),
(\out_{pre},\outt_{pre}),q_c)$ $.$ $c$ $.$$(q_c,(\inp_{post},\inpt_{post}),
(\out_{post},\outt_{post}),q_{F_\Pair})$ with $N+1$ cycles, where cycle $c:
(q_c,(\inp_c,\inpt_c),(\out_c,\outt_c),q_c)$  is the $N+1^{th}$ cycle. From the
induction hypothesis, we have:
$d(\out_{pre}.\out_{post},\outt_{pre}.\outt_{post}) \leq
Kd(\inp_{pre}.\inp_{post},\inpt_{pre}.\inpt_{post})$.
  
If $d(\inp_c,\inpt_c) = 0$, then
$d(\inp_{pre}.\inp_c.\inp_{post},\inpt_{pre}.\inpt_c,\inpt_{post}) =
d(\inp_{pre}.\inp_{post},\inpt_{pre}.\inpt_{post})$, and since $c$ does not
satisfy any of the conditions in \lemref{suffnonrob},
$d(\out_{pre}.\out_c.\out_{post},\outt_{pre}.\outt_c,\outt_{post}) =
d(\out_{pre}.\out_{post},\outt_{pre}.\outt_{post})$.  Hence, from the induction
hypothesis, it follows that $d(\out_\pi,\outt_\pi) \leq
Kd(\inp_\pi,\inpt_\pi)$.

If $d(\inp_c,\inpt_c) \neq 0$, then 
we have:
\begin{align*}
d(\out_\pi,\outt_\pi) & \leq d(\out_{pre}.\out_{post},\outt_{pre}.\outt_{post}) + 
			     d(\out_c,\outt_c) \\ 
		      & \leq K(d(\inp_{pre}.\inp_{post},\inpt_{pre}.\inpt_{post})) + Kd(\inp_c,\inpt_c) \\
                      & = Kd(\inp_\pi,\inpt_\pi).
\end{align*} 	
The first inequality follows from the definition of the Manhattan distance.
The second inequality follows from the induction hypothesis, and the
definitions of $K_{cyc}$ and $K$.  The last equality follows from the
definition of the Manhattan distance and the facts ---  $|\inp_{pre}| =
|\inpt_{pre}|$, $|\inp_{c}| = |\inpt_{c}|$ and $|\inp_{post}| = |\inpt_{post}|$.
\end{proof}

\IfKrobustThenRobust*

\begin{proof} First, observe that for the generalized Manhattan distance $\hd$,
there are constants $c, C$ such that for all $\s,\str$, $c \cdot d_1(\s,\str)
\leq \hd(\s,\str) \leq C \cdot d_1(\s,\str)$, where $d_1$ is the Manhattan
distance.  Thus, $\tran$ is not robust w.r.t. generalized Manhattan distances
$\dI, \dO$ iff $\tran$ is not robust w.r.t. Manhattan distances on the input
and the output words.  Thus, it suffices to focus on robustness of $\tran$
w.r.t. Manhattan distances.

(1): It follows from \lemref{suffnonrob} that $\tran$ is not robust 
w.r.t. Manhattan distances iff there exists a cycle satisfying Condition
$1$ of \lemref{suffnonrob}. Existence of such a cycle can be checked in 
\np. Therefore, robustness of $\tran$ can be checked in \conp.

(2): This follows directly from \lemref{suffnonrob}.
\end{proof}

\RobustnessdNon*
\begin{proof} {\bf [of (i)]}
To show undecidability, it suffices to consider transducers processing finite words. 
We show a reduction of the universality problem for $\fsum$-automata, 
with weights $\{-1,0,1\}$ and threshold $1$, to robustness of non-deterministic 
letter-to-letter transducers w.r.t. the Hausdorff set-similarity function. 

Let us fix a $\fsum$-automaton $\aut$ with alphabet $\{a,b\}$ and weights
$\{-1,0,1\}$. Without loss of generality, we assume that for every word $\s \in
\{a,b\}^*$, $\aut$ has a run of the value $|\s|$, \ie, a run in which each
transition taken by $\aut$ has weight $1$.  Let $\Sigma = \{a,b\} \times 
\{\bot, \top\}$,  let $\Gamma = \{a,b\} \times \{\bot, \top\} \times
\{\lett{-1}, \lett{0}, \lett{1}\}$.  To refer conveniently to words over
$\Sigma$ or $\Gamma$ observe that for $\s \in \{a,b\}^*$,  $\theta \in
\{\bot,\top\}^*$, $\tau \in \{\lett{-1}, \lett{0}, \lett{1}\}^*$, if $|\s| =
|\theta| = |\tau|$, then $\s \otimes \theta \otimes \tau$ is a word over
$\Gamma$ and $\s \otimes \theta$ is a word over $\Sigma$. 

We define generalized Manhattan distances $\dI$, $\dO$
the Hausdorff set-similarity function $\dsetO$, and
a nondeterministic letter-to-letter transducer $\tran$ such that $\tran$ is $1$-robust iff for
every word $\s \in \{a,b\}^*$, $\valueL{\aut}(\s) < 1$.  
The transducer $\tran$ is defined as
follows: for all $\s \in \{a,b\}^*$ and $\theta \in \{\bot^*,\top^*\}$,
$\funcDefinedBy{\trans}(\s \otimes \theta) = \{\s \otimes \theta \otimes \tau : \tau$ is a sequence
of weights of some run of $\aut$ on $\s\}$.
We define the generalized \lnorm $\dI$ as follows: for all $i,j \in \{a,b\}$
and $P,Q \in \{\bot,\top\}$ we have 
\begin{compactenum}
\item $\diffI( \lpair{i}{P},\eos) = \diffI(\eos,\lpair{i}{P}) = \infty$  
\item $\diffI(\lpair{i}{P},\lpair{j}{Q}) = 
\begin{cases}
\infty & \text{if } i \neq j\\
1 & \text{if } i = j, P \neq Q\\
0 & \text{otherwise}
\end{cases}$
\end{compactenum}
We define the generalized \lnorm $\dO$ as follows:  for all $i,j \in \{a,b\}$, 
$P,Q \in \{\bot,\top\}$ and $c,e \in \{\lett{-1},\lett{0},\lett{1}\}$ we have 
\begin{compactenum}
\item $\diffO(\langle i,P,\lett{c} \rangle,\eos) = \diffO(\eos,\langle i,P,\lett{c} \rangle) = \infty$
\item $\diffO(\langle i,P,\lett{c} \rangle ,\langle j,Q,\lett{e} \rangle) = 
\begin{cases}
\infty & \text{if } i \neq j\\
c + e & \text{if } i = j, P \neq Q\\
0 & \text{if } i=j, P = Q, c = e \\
1 & \text{if } i=j, P = Q, c \neq e \\
\end{cases}$
\end{compactenum}

\noindent Observe that for each $\s,\str \in \{a,b\}^*$, $\s \otimes\theta \in
\dom{\tran}$ iff $\theta \in \{\bot^{|\s|},\top^{|\s|}\}$  and
$\dI(\s \otimes \theta, \str \otimes \vartheta) < \infty$ iff $\s=\str$ (and $\theta, \vartheta \in \{\bot^{|\s|},\top^{|\s|}\}$).  
Further note that
$\dI(\s \otimes \theta,\s \otimes \theta)$ $=$ $0$ $=$
$\dsetO^H(\funcDefinedBy{\tran}(\s \otimes \theta),\funcDefinedBy{\tran}(\s \otimes \theta))$.
Thus, we only consider pairs of inputs of the form $\s \otimes \theta$ and
$\s \otimes \vartheta$, where $\theta \neq \vartheta$. 
Then, $\dI(\s\otimes \theta, \s \otimes \vartheta)
= |\s|$.  Let us denote the set $\{\tau: \tau$ is a sequence of weights of some
run of $\aut$ on $\s \}$ as $W$. Thus,
$\dsetO^H(\funcDefinedBy{\tran}(\s \otimes \theta),\funcDefinedBy{\tran}(\s \otimes \vartheta))$
equals $\sup_{\tau \in W} \, \inf_{\upsilon \in W}
\dO(\s \otimes \theta \otimes \tau, \s \otimes \vartheta \otimes \upsilon  )$, which in turn equals $\sup_{\tau \in
W} \inf_{\upsilon \in W} \fsum(\tau) + \fsum(\upsilon)$. Since for every $\s \in \{a,b\}^*$, $\aut$
has a run of the value $|\s|$, $\sup_{\tau \in W} \inf_{\upsilon \in W} \fsum(\tau) +
\fsum(\upsilon)$ equals the sum of $|\s|$ and the value of an accepting run of
$\aut$ on $\s$.  Therefore, $\trans$ is $1$-robust iff for every word $\s$ in
$\{a,b\}^*$, $\aut$ has an accepting run on $\s$ of value $<$ $1$.
\end{proof}


\begin{proof} {\bf [of (ii)]}
As before, it suffices to consider transducers processing finite words. 
We show a reduction of the universality problem for $\fsum$-automata, 
with weights $\{-1,0,1\}$ and threshold $1$, to robustness of non-deterministic 
letter-to-letter transducers w.r.t. the Inf-inf set similarity function. 

Let us fix a $\fsum$-automaton $\aut$ with alphabet $\{a,b\}$ and weights
$\{-1,0,1\}$.  Without loss of generality, we assume that for every word $\s
\in \{a,b\}^*$, $\aut$ has a run of the value $|\s|$, \ie, a run in which each
transition taken by $\aut$ has weight $1$.  Let $\Sigma = \{ \bot, a,b\}$ and
$\Gamma = \{ \bot, \lett{-1}, \lett{0}, \lett{1}\}$.  Let $\dI$, $\dO$ be the
generalized Manhattan distances and $\dsetO$ be the Inf-inf set-similarity
function. We define a transducer $\tran$ such that $\tran$ is $1$-robust iff
for every word $\s \in \{a,b\}^*$, $\valueL{\aut}(\s) < 1$.  $\tran$ is defined
as follows: for every $k > 0$, $\funcDefinedBy{\trans}(\bot^k) = \{\bot^k\}$
and for every $\s \in \{a,b\}^*$, $\funcDefinedBy{\trans}(\s) = \{ \tau : \tau$
is a sequence of weights of some run of $\aut$ on $\s \}$.

Let $\dI$ be a generalized \lnorm defined as follows:
for all $i,j \in \{a,b\}$ we have
\begin{compactenum}
\item $\diffI(i,\eos) = \diffI(\eos,i) = \infty$ 
\item $\diffI(i,j) =
\begin{cases}
0 & \text{if } i = j\\
1 & \text{otherwise}
\end{cases}$
\end{compactenum}
Let $\dO$ be a generalized \lnorm defined as follows: for all $i,j \in
\{\lett{-1},\lett{0},\lett{1}\}$ we have 
\begin{compactenum}
\item $\diffI(i,\eos) = \diffI(\eos,i) = \infty$
\item $\diffI(\bot,\eos) = \diffI(\eos,\bot) = \infty$
\item $\diffO(i, j) =
\begin{cases} 
0 & \text{if } i = j\\
1 & \text{if } i \neq j\\
\end{cases}$
\item $\diffO(\bot, i) = i + 1$ 
\end{compactenum} 

Now, observe that for every $k >0$,
$\dsetO^{\inf}(\funcDefinedBy{\trans}(\bot^k), \funcDefinedBy{\trans}(\bot^k))
= 0 = \dI(\bot^k, \bot^k)$, and for all $\s,\str \in \{a,b\}^*$ with $|\s|  = |\str|$,
$\dsetO^{\inf}(\funcDefinedBy{\trans}(\s), \funcDefinedBy{\trans}(\str)) = 0  \leq
\dI(\s,\str)$. Indeed, $1^{|\s|} \in \funcDefinedBy{\trans}(\s), 1^{|\str|} \in
\funcDefinedBy{\trans}(\str)$ and $|\s|  = |\str|$. We now consider the case when $\s
\in \{a,b\}^*$ and $\str = \bot^{|\s|}$.  Then, $\dI(\s,\str) = |\s|$ and
$\dsetO^{\inf}(\funcDefinedBy{\trans}(\s), \funcDefinedBy{\trans}(\str))$ equals
the sum of $|\s|$ and the minimal value of an accepting run of $\aut$ on $\s$
(notice that for $i \in \{\lett{-1},\lett{0},\lett{1}\}$, $\diffO(\bot, i) = 1
+ i$). Therefore, $\trans$ is $1$-robust iff for every word $\s$ in $\{a,b\}^*$, 
$\aut$ has an accepting run on $\s$ of value $<$ $1$.
\end{proof}


\end{document}